\documentclass[article, 11pt]{llncs}
\usepackage{amsmath,amssymb,amsbsy,amsfonts,latexsym, latexsym, amsopn,amstext,amsxtra,euscript,amscd, mathtools,braket,commath,color}
\usepackage{graphicx}
\usepackage{todonotes}
\usepackage{xspace}

\def\+{\oplus}

\def\cP{{\mathcal P}}

\def\cS{{\mathcal S}}

\def\aa{{\bf a}}
\def\bb{{\bf b}}

\def\uu{{\bf u}}

\def\xx{{\bf x}}
\def\yy{{\bf y}}

\def\00{{\bf 0}}
\def\11{{\bf 1}}

\def \F {{\mathbb F}}

\def \Z {{\mathbb Z}}

\def \Tr {{\rm Tr}_1^n}

\def \T {{\rm Tr}}

\def\wt{{\rm wt}}

\def\supp{{\rm supp}}

\def\dist{{\rm dist}}

\providecommand{\newoperator}[3]{%
  \newcommand*{#1}{\mathop{#2}#3}}

\newoperator{\FD}{\mathrm{FD}}{\nolimits}
\begin{document}

\title{Rayleigh quotients of Dillon's functions}

\author{Aditi Kar Gangopadhyay\inst{1}, Mansi\inst{1}, Bimal Mandal\inst{2}, 
\\ Aleksandr Kutsenko\inst{3}, Sugata Gangopadhyay\inst{4}
}
\institute{
Department of Mathematics 
\\ Indian Institute of Technology Roorkee, Roorkee 247667, India;\\
\email{\{aditi.gangopadhyay, mansi\}@ma.iitr.ac.in}
\and
Department of Mathematics\\
Indian Institute of Technology Jodhpur, 342030, India\\
\email{bimalmandal@iitj.ac.in} \and
Novosibirsk State University, Novosibirsk, 630090, Russia\\
\email{alexandrkutsenko@bk.ru}
\and
Department of Computer Science and Engineering,\\
Indian Institute of Technology Roorkee, Roorkee 247667, India;\\
\email{sugata.gangopadhyay@cs.iitr.ac.in}
}

\date{\today}
\maketitle
\thispagestyle{empty}

\abstract{The Walsh--Hadamard spectrum of a bent function uniquely determines a dual function. The dual of a bent function is also bent. A bent function that is equal to its dual is called a self-dual function. The Hamming distance between a bent function and its dual is related to its Rayleigh quotient. Carlet, Danielsen, Parker, and Sol\'e~\cite{CarletDPS10} studied Rayleigh quotients of bent functions in ${\mathcal PS}_{ap}$, and obtained an expression in terms of a character sum~\cite{DanielsenPS09}. We use another approach comprising of Desarguesian spreads to obtain the complete spectrum of Rayleigh quotients of bent functions in $\mathcal{PS}_{ap}$.}


\section{Introduction}
\label{intro}
Let $\F_{2}$ be the  prime field of characteristic $2$, and let $[n] = \{1, 2, \ldots, n\}$.  We denote by $\F_{2^{n}}$ the extension field of degree $n$ over $\F_{2}$.  The $n$-dimensional vector space of $n$-tuples of $\F_{2}$ is $\F_{2}^{n} = \{ (x_{1},x_2, \ldots, x_{n}): x_{i} \in \F_{2}, \mbox{ for all } i \in [n]\}$. We adopt the convention of writing a vector $\xx$ in $\F_{2}^{n}$ as $\xx = (x_1, x_2,\ldots, x_n)$ where $x_i \in \F_{2}$ are the coordinates. The weight of a vector $\xx\in\mathbb F_2^n$ is $\wt(\xx)=\displaystyle\sum_{i=1}^n x_i$ where sum is over integers. Let $E_e=\{\xx\in\mathbb F_2^n: \wt(\xx) \mbox{ is even}\}$ and $E_o=\{\xx\in\mathbb F_{2}^{n}: \wt(\xx) \mbox{ is odd}\}$. The inner product on $\F_{2}^{n}$ is defined by 
$\xx \cdot \yy = \oplus_{i=1}^{n} x_{i}y_{i}$, for all  $\xx, \yy \in \F_{2}^{n}$,
where $\+$ denotes addition over modulo $2$.

The trace function $\Tr : \F_{2^{n}} \rightarrow \F_{2}$ is defined by  
$\Tr(x) = x+x^{2}+\cdots+x^{2^{n-1}}$, for all  $x \in \F_{2^{n}}$. The inner product on $\F_{2^{n}}$ is defined by 
$(x, y) \mapsto \Tr(xy)$,  for all  $x, y \in \F_{2^{n}}$. We use one of the above inner products depending on our choice of the domain of Boolean functions. Suppose $U$ is a subspace of $\F_{2}^{n}$. The dual space of $U$ is
\begin{equation*}
U^{\perp} = \{ \xx \in \F_{2}^{n} : \xx \cdot \yy = 0, \mbox{ for all } \yy \in U\},
\end{equation*}
or, equivalently, if $U$ is considered to be a vector subspace of $\F_{2^{n}}$,
the dual space of $U$ is 
\begin{equation*}
U^{\perp} = \{ x \in \F_{2^{n}} : \Tr(xy) = 0, \mbox{ for all } y \in U\}.
\end{equation*}
\begin{definition}
    A Boolean function in $n$ variables is an $\F_{2}$-valued function from $\F_{2}^{n}$, or equivalently from $\F_{2^{n}}$ to $\F_{2}$.
\end{definition}

Carlet~\cite{carlet2021}, and Cusick and St\u anic\u a~\cite{Pante} are excellent sources for studying Boolean functions relevant to cryptography and coding theory. Any Boolean function $f$ in $n$ variables has a unique polynomial representation, called the algebraic normal form (ANF) of $f$, of the form 
\begin{equation*}
f(x_1, x_2,\ldots, x_n) = \displaystyle\bigoplus_{\aa \in \F_{2}^{n}} \mu_{\aa} \prod_{i=1}^{n}x_i^{a_i},
\end{equation*}
where the coefficients $\mu_{\aa} \in \F_{2}$,  for all  $\aa = (a_1, a_2,\ldots, a_n) \in \F_{2}^{n}$. Let $\mathfrak{B}_n$ be the set of all Boolean functions in $n$ variables. The weight of $f$ is $\wt(f)=|\supp(f)|$, where  $\supp(f)= \{\xx \in \F_2^n: f(\xx) \neq 0\}$ is called the support of $f$. A function $f\in\mathfrak B_n$ is said to be balanced if $\wt(f)=2^{n-1}$. The algebraic degree of $f$ is $deg(f) = \max\{\wt(\aa) : \mu_{\aa} \neq 0\}$. A Boolean function of degree atmost $1$ is called an affine function. The set of all $n$-variable affine Boolean functions is denoted by $\mathfrak A_n$. Affine functions in $n$ variables have the form $f(x_1, x_2,\ldots, x_n) = \oplus_{i=1}^n a_i x_i\oplus \varepsilon$, where $a_i, \varepsilon \in \F_{2}$, for all $i = 1, 2,\ldots, n$. If $\varepsilon = 0$, then $f$ is said to be an $n$-variable linear function, and the set of all such functions is denoted by $\mathfrak{L}_n$.  The Hamming distance between two Boolean functions $f, g \in \mathfrak{B}_n$ is $dist(f, g) =|\{ \xx \in \F_{2}^{n} : f(\xx) \neq g(\xx) \}|$. The minimum of the Hamming distances of $f$ from the set of affine functions is called its nonlinearity, i.e., $nl_{f} = dist(f, \mathfrak A_n) =  \min \{dist(f,g): g \in \mathfrak A_n\}$. The Walsh--Hadamard transform of $f \in \mathfrak B_n$ is an integer-valued function $W_{f} : \F_{2}^{n} \rightarrow \Z$, defined by 
\begin{equation*}
W_{f}(\uu) = \displaystyle\sum_{\xx \in \F_{2}^{n}} (-1)^{f(\xx) \+ \uu \cdot \xx},
\end{equation*}
for all $\uu \in \F_{2}^{n}$. The multiset $[ W_{f}(\uu): \uu \in \F_{2}^{n}]$ is said to be the Walsh--Hadamard spectrum of $f$. For each  $\uu \in \F_{2}^{n}$, $W_{f}(\uu)$ is used to calculate the distance between $f$ and linear function $\uu\cdot\xx$. The relation between the nonlinearity of $f$ and its Walsh--Hadamard spectrum is given by 
\begin{equation*}
    nl_{f} = 2^{n-1} - \frac{1}{2} \max_{\uu \in \F_2^n} |W_{f}(\uu)|.
\end{equation*}  

The functions with high nonlinearity have been studied in great detail due to their applications in designing several cryptographic primitives. For any positive, even integer $n$, the upper bound of nonlinearity is attained by a class of functions called bent functions. Mesnager~\cite{Mesnager} is the most extensive reference for studying bent functions. Contrarily, Tokareva~\cite{Tokareva} provides a concise review of the theory of bent functions along with several intriguing historical insights. This paper centres around finding the distribution of Hamming distance of a bent function and their dual for some special classes of bent functions.

\subsection{Bent functions, duality, and Rayleigh quotients}
Any Boolean function $f\in\mathfrak{B}_n$ satisfies Parseval's identity $\sum_{\uu\in\F_{2}^{n}}|W_{f}(\uu)|^2 = 2^{2n}$. In other words, the average of the squares of the Walsh--Hadamard coefficients is $2^{n}$. Since the maximum value in a data set of non--negative integers is always greater than or equal to the average, i.e., $\max_{\uu \in \F_{2}^{n}}  |W_{f}(\uu)|\geq 2^{\frac{n}{2}}$. Thus, nonlinearity of any Boolean function is bounded above by $2^{n-1}-2^{\frac{n}{2}-1}$, i.e., for any $f\in\mathfrak B_n$
\begin{equation*}
nl_{f} \leq 2^{n-1} - 2^{\frac{n}{2}-1}.
\end{equation*}
When $n$ is even, this bound is attainable by a class of functions which are called bent Boolean functions.
\begin{definition}
\label{bentDef}
Let $n$ be a positive even integer. A Boolean function $f \in \mathfrak{B}_n$ that has nonlinearity 
$nl_{f} = 2^{n-1} - 2^{\frac{n}{2}-1}$ is said to be bent. 
\end{definition}
Equivalently, a Boolean function in $n$ variables is called bent if and only if 
$|W_{f}(\uu)| = 2^{\frac{n}{2}}$, for all $\uu\in \F_{2}^{n}$. Thus, bent functions are defined over an even number of variables and its Walsh--Hadamard spectrum consists of two values only.
\begin{definition}
If $f \in \mathfrak{B}_n$ is a bent function, then there exists a unique function $\widetilde{f}: \F_{2}^{n} \rightarrow \F_{2}$ such that 
\begin{equation*}
W_{f}(\uu) = (-1)^{\widetilde{f}(\uu)} 2^{\frac{n}{2}}, \mbox{ for all } \uu \in \F_{2}^{n}. 
\end{equation*}
The function $\widetilde{f}$ is said to be the dual of $f$, and it is also a bent function in $n$ variables. 
\end{definition}
In particular, if $f=\widetilde{f}$ then $f$ is called self-dual, and if $f = \widetilde{f} \+ 1$, i.e., $\widetilde{f}$ is complement of $f$ then $f$ is called anti-self-dual.
Rothaus~\cite{Rothaus76} first constructed a class of bent functions, called Maiorana--McFarland bent functions. Let $n=2k$. A function $f:\F_{2}^{k} \times \F_{2}^{k} \rightarrow \F_{2}$ is in Maiorana--McFarland bent class  if 
\begin{equation*}
f(\xx, \yy) = \xx\cdot \pi(\yy) \+ g(\yy),
\end{equation*}
for all $\xx, \yy \in \F_{2}^{k}$, where $\pi:\F_{2}^{k} \rightarrow \F_{2}^{k}$ is permutation and $g: \F_{2}^{k}\rightarrow \F_{2}$ is any Boolean function in $k$ variables. The dual of Maiorana--McFarland function $f$ is also of the same type given as:
\begin{equation*}
\widetilde{f}(\xx, \yy) = \xx \cdot \pi^{-1}(\yy)\+g(\pi^{-1}(\yy)),
\end{equation*}
for all $\xx, \yy \in \F_{2}^{k}$.
Dillon~\cite{Dillon72,Dillon74} constructed another class of bent functions, called partial spread ($\mathcal{PS}$), whose support is the union of elements of partial spreads of $\F_{2}^{n}$. Further, Dillon constructed a subclass of partial spread functions consisting of elements from a spread. We call these functions Dillon's functions~\cite{DanielsenPS09}. 
\begin{definition}
Let $n = 2k$ be an even integer. A partial spread of $\F_{2}^{n}$ consists of pairwise supplementary $k$-dimensional subspaces of $\F_{2}^{n}$. A partial spread is a full spread if the union of its elements equals $\F_{2}^{n}$.
\end{definition}
The $\mathcal{PS}$ class is divided into two disjoint class of functions, called $\mathcal{PS^-}$ and $\mathcal{PS^+}$, i.e., $\cP\cS = \cP\cS^{-} \cup \cP\cS^{+}$. We first define the indicator function of a set $S$ as 
\[\phi_S(\xx)=
\left\{\begin{array}{ll}
1 & \mbox{ if } \xx \in S;\\
0 & \mbox{ if } \xx \notin S.
\end{array}\right. \]
Let $U$ and $V$ be two subspaces of $\F_{2}^{n}$. We refer to $U$ and $V$ as ``disjoint'' if $ U \cap V= \{ \00 \}$, where $\00$ is the all-zero vector of $\F_{2}^{n}$. Let $E_{i}$, $i = 1,2, \ldots, 2^{k-1}$ are mutually ``disjoint'' $k$-dimensional subspaces of $\F_{2}^{n}$. A bent function $f \in \mathfrak{B}_n$ is said to be in $\cP\cS^{-}$ class if it is 
of the form
\begin{equation*}
	f(\xx) = \displaystyle\sum_{i = 1}^{2^{k-1}} \phi_{E_{i}}(\xx) - 2^{k-1} \phi_{\{\00\}}(\xx).
\end{equation*}
The dual of $f \in \cP\cS^{-}$ is of the form
	$\widetilde{f}(\xx) = \displaystyle\sum_{i = 1}^{2^{k-1}} \phi_{E_{i}^{\perp}}(\xx) - 2^{k-1} \phi_{\{\00\}}(\xx)$.
A bent function $f \in \mathfrak{B}_n$ is said to be in $\cP\cS^{+}$ class if it is 
of the form
\begin{equation*}
	f(\xx) = \displaystyle\sum_{i = 1}^{2^{k-1}+1} \phi_{E_{i}}(\xx) - 2^{k-1} \phi_{\{\00\}}(\xx),
\end{equation*}
where $E_{i}$, $i = 1,2, \ldots, 2^{k-1}+1$ are mutually ``disjoint'' $k$-dimensional subspaces of $\F_{2}^{n}$. The dual of $f \in \cP\cS^{+}$ is of the form
\begin{equation*}
	\widetilde{f}(\xx) = \displaystyle\sum_{i=1}^{2^{k-1}+1} \phi_{E_{i}^{\perp}}(\xx) 
	- 2^{k-1} \phi_{\{\00\}}(\xx). 
\end{equation*}
The distance between a bent function and its dual is an interesting research problem. Danielsen et al. \cite{DanielsenPS09} introduced the normalized Rayleigh quotient of a bent Boolean function that measures the distance of the bent function to its dual. 
\begin{definition}
\label{rayleighDefn1}
The Rayleigh quotient $S_{f}$ of a Boolean function $f$ in
$n$ variables is
\begin{equation*}
S_{f} = \displaystyle\sum_{\xx,\yy \in \F_{2}^{n}} (-1)^{f(\xx)\+f(\yy)\+\xx \cdot \yy}
= \displaystyle\sum_{\xx\in\F_{2}^{n}} (-1)^{f(\xx)} W_{f}(\xx).
\end{equation*}
\end{definition}
For bent functions, we use the normalized version of the Rayleigh quotient. The normalized Rayleigh quotient $N_{f}$ of a bent function $f$ in $n$ variables is 
\begin{equation*}
N_{f} = \displaystyle\sum_{\xx \in \F_{2}^{n}} (-1)^{f(\xx) \+ \widetilde{f}(\xx)} =2^{-\frac{n}{2}}S_{f}.
\end{equation*}
For an $n$-variable bent Boolean function $f$, $N_{f} = 2^{n}$ is $f$ is self-dual, and $N_{f} = -2^{n}$ if $f$ is anti-self-dual. 
The Hamming distance between a bent function $f$ and its dual $\widetilde{f}$ is
\begin{align*}
\dist(f,\widetilde{f}) 
= | \{\xx \in \F_{2}^{n} : f(\xx) \neq \widetilde{f}(\xx)  \}| 
&= 2^{n-1} - \frac{N_{f}}{2} \\
&=\left\{\begin{array}{ll} 
	0 & \mbox{ if } f \mbox{ is self-dual}\\
	2^{n} & \mbox{ if } f \mbox{ is anti self-dual}
\end{array}\right.
\end{align*}
Now, we are going to discuss some results related to Rayleigh's quotients of symmetric bent functions.
\subsection{Distribution of Rayleigh quotients of symmetric bent functions}
A Boolean function $f \in \mathfrak{B}_{n}$ is called a symmetric Boolean function if its output remains the same for any permutation of its inputs. In 1994, Savick\'y \cite{SAVICKY1994407} derived the following characterization for symmetric bent functions.
\begin{theorem}\textup{\cite[Theorem 3.3]{SAVICKY1994407}}
\label{thm-symm}
If $f(\xx)=c_k$, $\wt(\xx)=k$, is a symmetric function of an even number of variable $n$, then the following statements are equivalent:
\begin{enumerate}
    \item The function $f$ is a bent function;
    \item For all $k = 0,1, \dots, n-2$, the identity $c_{k+2}=c_{k} \+ 1$ is satisfied;
    \item There exist $\varepsilon_1,\varepsilon_2\in\F_{2}$, $f(\xx)=\oplus_{1\leq i<j\leq n} x_ix_j \oplus \varepsilon_1(\oplus_{i=1}^n x_i)\oplus \varepsilon_2$, for all $\xx\in\F_{2}^{n}$.
\end{enumerate}
\end{theorem}
We first derive the dual of symmetric bent functions explicitly. Then the possible Rayleigh quotient of symmetric bent functions are calculated. It is observed that the dual of a symmetric bent function is either self-dual or anti-self-dual or exactly half of the output of dual function is the same as the output of the function and the other half is different. These cases are depended on the number of variables and the coefficient related to the constant and linear terms. The results are given in the following Propositions.
\begin{proposition}
Let $f \in \mathfrak{B}_{n}$ be a symmetric bent function then the dual of $f$ is
\begin{eqnarray*}
    \tilde{f}(\xx) = \begin{cases}
    \wt(\xx) \+ c_{\wt(\xx)} \+ \frac{n}{4} & \text{if $\frac{n}{2}$ is even}\\
    & \\
    \begin{cases}
      \wt(\xx) \+ c_{\wt(\xx)+1}\+\lfloor \frac{n}{4} \rfloor & \text{if $\wt(\xx)\leq n-1$}\\
      n \+ 1 \+ c_{1} \+ \lfloor \frac{n}{4} \rfloor & \text{if $\wt(\xx)=n$}
    \end{cases}& \text{if $\frac{n}{2}$ is odd}
    \end{cases}.
    \end{eqnarray*}
\end{proposition}

\begin{proposition}
Let $f\in \mathfrak{B}_{n}$ be a symmetric bent Boolean function, then the Rayleigh quotient of $f$ is
\begin{equation*}
    N_{f} = \begin{cases}
      0 & \text{if $\frac{n}{2}$ is even}\\
      & \\
      \begin{cases}
      2^{n} & \begin{cases} \text{either $\lfloor \frac{n}{4} \rfloor$ is even and $c_{0}=c_{1}$}\\
      \text{or $\lfloor \frac{n}{4} \rfloor$ is odd and $c_{0}=c_{1} \+ 1$}
      \end{cases}\\
      & \\
      -2^{n} & \begin{cases} \text{either $\lfloor \frac{n}{4} \rfloor$ is odd and $c_{0}=c_{1}$}\\
      \text{or $\lfloor \frac{n}{4} \rfloor$ is even and $c_{0}=c_{1}+1$}
      \end{cases}
      \end{cases} & \text{if $\frac{n}{2}$ is odd}
    \end{cases},
    \end{equation*}
    where $c_{0}$ and $c_{1}$ are the output value of function for input of weight $0$ and $1$, respectively, and $\lfloor \cdot \rfloor$ is the floor function.
\end{proposition} 
We refer \cite{quad,SAVICKY1994407} to get details of these propositions. From above, we have a complete spectrum of Rayleigh quotients for symmetric bent Boolean functions. By ease, without doing many calculations, if we are given a symmetric bent Boolean function we can find out its Rayleigh quotient. But it is not the case with Dillon's functions. The method of finding the Rayleigh quotient of Dillon's function is given in~\cite{CarletDPS10}. We are interested in finding the possible distance between Dillon's functions and their duals using another approach. So, we partition the bent functions into distinct sets based on the distance between them and their duals. At this point, it is worth noting that Kolomeec~\cite{Kolomeec17} proved that the minimum Hamming distance between two bent functions is $2^{\frac{n}{2}}$. Therefore, the largest possible Hamming distance between two bent functions is $2^{n} - 2^{\frac{n}{2}}$, except for anti-self-dual bent functions.

\section{Properties of bent functions and their duals}
In this section, we study some metrical properties of bent functions and their duals. These expressions can help us characterise the distance between a bent function and its dual. Let us denote $E_e=\{\xx\in\F_{2}^{n}: \wt(\xx) \mbox{ is even}\}$ and $E_o=\{\xx\in\F_{2}^{n}: \wt(\xx) \mbox{ is odd}\}$. The restricted Walsh--Hadamard transform over $A\subset \F_{2}^{n}$ of a Boolean function $f\in\mathfrak{B}_{n}$ at $\aa\in\F_{2}^{n}$ is denoted by $W_{f}(\aa)|_A$, defined as $W_{f}(\aa)|_A=\displaystyle\sum_{\xx\in A}(-1)^{f(\xx) \+ \aa \cdot \xx}$. 
\begin{proposition}
\label{metricTh1}
Let $f\in\mathfrak{B}_{n}$, where $n = 2k\geq 4$, be a bent function and its dual
$\widetilde{f}$. The distance between $f$ and $\widetilde{f}$ can be written as  
\begin{equation*}
\begin{split}
\dist(f,\widetilde{f})&=2^{n-1}-2^{k-1}(-1)^{f(\00)}+\frac{1}{2^k}\displaystyle\sum_{\uu\in\supp(f)}W_{f}(\uu)\\
&=2^{n-1}-\frac{1}{2^{k+1}}\displaystyle\sum_{\uu\in\F_{2}^{n}}W_{D_{\uu}f}(\uu)+\frac{1}{2^k}\displaystyle\sum_{\uu\in\F_{2}^{n}} W_{D_{\uu}f}(\uu)|_{E_o}.
\end{split}
\end{equation*}
\end{proposition}

\begin{proof}
Let $f\in\mathfrak{B}_n$ be a bent function, $n=2k\geq 4$, and $W_{f}(\uu)=2^{k}(-1)^{\widetilde{f}(\uu)}$, for all $\uu\in\F_{2}^{n}$. The Hamming distance between $f$ and $\widetilde{f}$ is
\allowdisplaybreaks[4]
\begin{align*}
\dist(f,\widetilde{f})&=|\{\xx\in\F_{2}^{n}: f(\xx)\neq \widetilde{f}(\xx)\}|
=2^{n-1}-\frac{1}{2}\displaystyle\sum_{\xx\in\F_{2}^{n}}(-1)^{f(\xx) \+ \widetilde{f}(\xx)}\\
&= 2^{n-1}-\frac{1}{2}\displaystyle\sum_{\xx\in\F_{2}^{n}}(1-2f(\xx))(-1)^{\widetilde{f}(\xx)}\\
&= 2^{n-1}-\frac{1}{2}\displaystyle\sum_{\xx\in\F_{2}^{n}}(-1)^{\widetilde{f}(\xx)}
+\displaystyle\sum_{\xx\in\F_{2}^{n}}f(\xx)(-1)^{\widetilde{f}(\xx)}\\
&= 2^{n-1}-\frac{1}{2^{k+1}}\displaystyle\sum_{\xx\in\F_{2}^{n}}W_{f}(\xx)
+\frac{1}{2^k}\displaystyle\sum_{\xx\in\supp(f)}W_{f}(\xx)\\
&= 2^{n-1}-\frac{1}{2^{k+1}}\displaystyle\sum_{\xx\in\F_{2}^{n}}\displaystyle\sum_{\yy\in\F_{2}^{n}}(-1)^{f(\yy) \+ \yy\cdot \xx} +\frac{1}{2^k}\displaystyle\sum_{\xx\in\supp(f)}W_{f}(\xx)\\
&= 2^{n-1}-\frac{1}{2^{k+1}}\displaystyle\sum_{\yy\in\F_{2}^{n}}(-1)^{f(\yy)}
\displaystyle\sum_{\xx\in\F_{2}^{n}}(-1)^{\yy\cdot \xx}
+\frac{1}{2^k}\displaystyle\sum_{\xx\in\supp(f)}W_{f}(\xx)\\
&= 2^{n-1}-2^{k-1}(-1)^{f(\00)}+\frac{1}{2^k}\displaystyle\sum_{\xx\in\supp(f)}W_{f}(\xx), 
\end{align*}
and 
\allowdisplaybreaks
\begin{align*}
\dist(f,\widetilde{f})&=2^{n-1}-\frac{1}{2}\displaystyle\sum_{\xx\in\F_{2}^{n}}(-1)^{f(\xx) \+\widetilde{f}(\xx)}=2^{n-1}-\frac{1}{2^{k+1}}\displaystyle\sum_{\xx\in\F_{2}^{n}}(-1)^{f(\xx)}W_{f}(\xx)\\
&=2^{n-1}-\frac{1}{2^{k+1}}\displaystyle\sum_{\xx\in\F_{2}^{n}}
\displaystyle\sum_{\yy\in\F_{2}^{n}}(-1)^{f(\xx) \+ f(\yy) \+ \yy\cdot \xx}\\
&=2^{n-1}-\frac{1}{2^{k+1}}\displaystyle\sum_{\xx\in\F_{2}^{n}}
\displaystyle\sum_{\yy\in\F_{2}^{n}}(-1)^{f(\xx)\+f(\yy)}(-1)^{\yy\cdot \xx}\\
&=2^{n-1}-\frac{1}{2^{k+1}}\displaystyle\sum_{\xx\in\F_{2}^{n}}
\displaystyle\sum_{\uu\in\F_{2}^{n}}(-1)^{f(\xx)\+f(\xx\+\uu)}(-1)^{(\xx\+\uu)\cdot \xx}\\
&=2^{n-1}-\frac{1}{2^{k+1}}\displaystyle\sum_{\uu\in\F_{2}^{n}}\displaystyle\sum_{\xx\in\F_{2}^{n}} (-1)^{D_{\uu}f(\xx)}(-1)^{\wt(\xx)}(-1)^{\uu\cdot \xx}\\
&=2^{n-1}-\frac{1}{2^{k+1}}\displaystyle\sum_{\uu\in\F_{2}^{n}}\bigg\{\displaystyle\sum_{\xx\in E_e} (-1)^{D_{\uu}f(\xx) \+ \uu\cdot \xx}-\displaystyle\sum_{\xx\in E_o} (-1)^{D_{\uu}f(\xx) \+ \uu\cdot \xx}\bigg\}\\
&=2^{n-1}-\frac{1}{2^{k+1}}\displaystyle\sum_{\uu\in\F_{2}^{n}}\bigg\{\displaystyle\sum_{\xx\in \F_{2}^{n}} (-1)^{D_{\uu}f(\xx) \+ \uu\cdot \xx}-2\displaystyle\sum_{\xx\in E_o} (-1)^{D_{\uu}f(\xx) \+ \uu\cdot \xx}\bigg\}\\
&=2^{n-1}-\frac{1}{2^{k+1}}\displaystyle\sum_{\uu\in\F_{2}^{n}}W_{D_{\uu}f}(\uu)+
\frac{1}{2^k}\displaystyle\sum_{\uu\in \F_{2}^{n}} W_{D_{\uu}f}(\uu)|_{E_o}. 
\end{align*}
\end{proof}

\begin{corollary}
For any bent function $f \in \mathfrak{B}_n$ where $n \geq 4$
is a positive even integer, we have
\begin{equation*}
2\displaystyle\sum_{\uu\in\supp(f)}W_{f}(\uu)+\displaystyle\sum_{\uu\in\F_{2}^{n}}W_{D_{\uu}f}(\uu)-2\displaystyle\sum_{\uu\in\F_{2}^{n}}W_{D_{\uu}f}(\uu)|_{E_o}-2^{n}(-1)^{f(\00)}=0.
\end{equation*}
\end{corollary}
The set of bent functions can be partitioned using the distance between bent functions and their duals. Let us define the equivalence classes of bent functions.
\begin{definition}
	Two bent functions $f,g\in\mathfrak{B}_{n}$ are called distance equivalent if  $\dist(f,\widetilde{f})=\dist(g,\widetilde{g})$, and are denoted by $f\sim g$. 
\end{definition}

For any bent functions $f,g,h\in\mathfrak B_{n}$, $f\sim f$, and if $f\sim g$ then $g\sim f$, and if $f\sim g$ and $g\sim h$ then $f\sim h$. Thus, ``$\sim$'' is an equivalence relation on the set of bent functions. Let us define equivalence classes of bent functions corresponding to this equivalence relation as \begin{equation*}
	cl(i)=\{f\in \mathfrak B_{n}: f \mbox{ is bent and } \dist(f,\widetilde{f})=i\}, \text{ for all } 0\leq i\leq 2^{n}.
\end{equation*}

Equivalence classes of self-dual and anti-self-dual functions are $cl(0)$ and $cl(2^{n})$, respectively. Except for the anti-self-dual function, $\dist(f,\widetilde{f}) \leq 2^{n}-2^{\frac{n}{2}}$. For every $0\leq i\neq j\leq 2^{n}, cl(i)\cap cl(j)=\emptyset$. Thus, the set of all bent functions in $n$ variables is $\cup_{i=0}^{2^{n}} cl(i)$. Here, the two obvious questions arise, as below:
\begin{itemize}
		\item Are the bent functions of the same class, i.e., $cl(i)$, for a fixed $0\leq i\leq 2^{n}$, affine equivalent? Or is it possible to find a non--singular affine transformation such that both bent functions belong to the same class?
	\item What are the possible distances for a particular class of bent functions? 
\end{itemize}

We are going to address both questions in this paper. Now, we prove that these distance classes are, in general, not affine invariant. In the following section, we determine the specifics of potential classes of a class of bent functions.

Let $f$ be a bent function in $n=2k$ variables and $g(\xx)=f(\xx A \+ \bb)$, for all $\xx\in\F_{2}^{n}$, where $A$ is a non--singular binary matrix of order $n$ and $\bb\in\F_{2}^{n}$. Then $g(\00)=f(\bb)$, $\supp(g)=\supp(f)B\+\bb B=\{\xx B\+\bb B: f(\xx)=1\}$, where $B=A^{-1}$, and for any $\uu\in\F_{2}^{n}$,
\begin{equation*}
\begin{split}
	W_{g}(\uu)&=\displaystyle\sum_{\xx\in\F_{2}^{n}} (-1)^{f(\xx A \+\bb) \+ \uu\cdot \xx}= \displaystyle\sum_{\yy\in\F_{2}^{n}} (-1)^{f(\yy) \+ \uu\cdot (\yy B \+ \bb B)}\\
	&=(-1)^{\uu B^{T}\cdot \bb}\displaystyle\sum_{\yy\in\F_{2}^{n}} (-1)^{f(\yy) \+ \uu B^{T}\cdot \yy}=(-1)^{\uu B^{T}\cdot \bb} W_{f}(\uu B^{T}).
\end{split}
\end{equation*}
Thus,
\begin{equation*}
\allowdisplaybreaks
\begin{split}
	\dist(g,\widetilde{g})&=2^{n-1}-2^{k-1}(-1)^{g(0)}+\frac{1}{2^{k}} \displaystyle\sum_{\uu\in \supp(g)}W_{g}(\uu)\\
	&=2^{n-1}-2^{k-1}(-1)^{f(\bb)}+\frac{1}{2^{k}}\displaystyle\sum_{\uu\in \supp(f)B+\bb B}(-1)^{\uu B^{T}\cdot \bb}W_{f}(\uu B^{T})\\
	&=2^{n-1}-2^{k-1}(-1)^{f(\bb)}\\
	& +\frac{1}{2^{k}}\displaystyle\sum_{\uu'\in \supp(f)}(-1)^{(\uu' \+ \bb)BB^{T}\cdot \bb} W_{f}((\uu' \+ \bb)BB^{T})\\
	&\neq  \dist(f,\widetilde{f}), \text{ in general}.
\end{split}
\end{equation*}

Let $A$ be an orthogonal binary matrix. Then $BB^{T}=(A^{T}A)^{-1}=I_{n}$, and 
\begin{equation*}
	\dist(g,\widetilde{g})=2^{n-1}-2^{k-1}(-1)^{f(\bb)}+\frac{(-1)^{\wt(\bb)}}{2^k}\displaystyle\sum_{\uu'\in \supp(f)}(-1)^{\uu'\cdot b}W_{f}(\uu' \+ \bb).
\end{equation*}

In particular if $\bb=\00$, then $\dist(g,\widetilde{g})=\dist(f,\widetilde{f})$. We get the next results directly from the above expressions.
\begin{theorem}
	Let $f\in\mathfrak B_n$ be bent and $g(\xx)=f(\xx A \+ \bb)$, for all $\xx\in\F_{2}^{n}$, where $A$ is a binary orthogonal matrix of order $n$ and $\bb\in\F_{2}^{n}$. If $\bb=\00$, then it preserves the distance, and if $\bb\neq \00$, then  $\dist(f,\widetilde{f}) =\dist(g,\widetilde{g})$ if and only if one of the  following condition holds.
	\begin{itemize}
		\item[1.] If $f(\00)=f(\bb)$, then $\displaystyle\sum_{\uu\in \supp(f)} W_{f}(\uu)=(-1)^{\wt(\bb)} \displaystyle\sum_{\uu\in \supp(f)} (-1)^{\bb\cdot \uu} W_{f}(\uu \+ \bb)$.
		\item[2.] If $f(\00)\neq f(\bb)$, then $\displaystyle\sum_{\uu\in \supp(f)}  W_{f}(\uu)-(-1)^{\wt(\bb)} \displaystyle\sum_{\uu\in \supp(f)} (-1)^{\bb\cdot \uu} W_{f}(\uu \+ \bb)=2^{2k}(-1)^{f(\00)}$. 
	\end{itemize}
\end{theorem} 

From the above theorem, it is clear that a self-dual bent function may not be a self-dual bent after a suitable non--singular affine transformation. Similarly, a suitable non--singular affine transformation can be used to construct a self-dual bent function from a bent function that is not self-dual.\\ For example let us consider $n=4$, $f(\xx)=x_1x_3+x_2x_4$, for all $\xx\in\F_{2}^{4}$, and two binary non--singular matrices $A$ and $A'$ of order $4$ as
\begin{equation*}
	A=
\begin{bmatrix} 
	1 & 1 & 1 & 0 \\
	1 & 1 & 0 & 1 \\
	1 & 0 & 1 & 1 \\
	0 & 1 & 1 & 1 \\
\end{bmatrix},
\quad
\;\;
	A'=
\begin{bmatrix} 
	1 & 0 & 0 & 0 \\
	1 & 1 & 0 & 0 \\
	1 & 1 & 1 & 0 \\
	1 & 1 & 1 & 1 \\
\end{bmatrix}.
\quad
\end{equation*}

Here $f$ is a self-dual bent function and $A$ is an orthogonal matrix, but $A'$ is not. Suppose $g(\xx)=f(\xx A)$, $g'(\xx)=f(\xx A')$, and $h(\xx)=f(\xx A \+(0,1,1,1))$ for all $\xx\in\F_{2}^{4}$. Then $g$ is also self-dual bent, but $g'$ and $h$ are not self-dual. Here $\dist(g',\widetilde{g'})=8$ and $\dist(h,\widetilde{h})=8$.
\section{Distribution of Rayleigh quotients of Dillon's functions}
Partial spread class of bent functions, denoted by $\cP\cS$, was constructed by Dillon~\cite{Dillon72,Dillon74}. Dillon also constructed an important subclass of $\cP\cS$ referred as $\cP\cS_{ap}$, where ``$ap$'' stands for an affine plane.

\begin{definition}
Let $x, y \in F_{2^k}$, where $k=\frac{n}{2}$. The class denoted by $\mathcal{PS}_{ap}$ in~\textup{\cite{CarletDPS10}} consists function of the type $f(x, y) = g(x/y)$ with the convention that $x/y = 0$ if $y = 0$, and where $g$ is balanced Boolean function and $g(0) = 0$.
\end{definition}

\begin{theorem} 
\label{thmco} 
\textup{\cite[Theorem 4.5]{CarletDPS10}} A Boolean function $f$ as defined above is self-dual bent if $g$ satisfies $g(1) = 0$, and, for all $u \neq 0$ the relation $g(u) = g(1/u)$. There are exactly $\binom{2^{\frac{n}{2}-1}-1}{2^{\frac{n}{2}-2}}$ such functions.
\end{theorem}
Danielsen, Parker, and Sol\'e~\cite{DanielsenPS09} considered the Rayleigh quotient of Dillon's functions and obtained an expression in terms of character sums.
The character sum for these functions is defined in~\cite{DobbertinL08} as
\begin{equation}
	\label{chr_sum}
	K_{g}:=\displaystyle\sum_{u}(-1)^{g(u)+g(1/u)}.
\end{equation}

If $g=Tr_{1}^{k}$ then $K_{g}$ is Kloosterman sum.
\begin{theorem}
	\label{thm_rq}
	\textup{\cite[Theorem7]{DanielsenPS09}} Let $f$ be a bent function constructed from a Dillon $g$ as above, then its $N_{f} = 2^{k} + (2^{k-1})K_{g}.$
\end{theorem}

In this section, we derive the expression of the distance between partial spread bent functions and their duals. In the rest of the paper, we consider $n=2k\geq 4$ and $f \in \mathfrak{B}_n$ be a bent  function with its dual denoted by $\widetilde{f}$.

\subsection{Distances between $\cP\cS^{+}$ bents and their duals}
Let $f \in \mathfrak{B}_n$ be any $\cP\cS^{+}$ class bent function such that $\supp(f)=\cup_{i=1}^{2^{k-1}+1}E_{i}$, where $E_{i}$'s are ``disjoint'' $k$-dimensional subspaces of $\F_{2}^{n}$. If the dual of $f$ is $\widetilde{f}$, then $\supp(\widetilde{f})=\cup_{i=1}^{2^{k-1}+1}E_{i}^{\perp}$. It is also clear that $f(\00)=\widetilde{f}(\00)=1$. 
\begin{theorem}
\label{cor-ps+}
Let $f \in \mathfrak{B}_n$, $n=2k\geq 4$, be a $\cP\cS^{+}$ class bent function with its dual $\widetilde{f}$ and $\supp(f)=\cup_{i=1}^{2^{k-1}+1}E_{i}$, where $E_{i}$'s are defined as above. Then
\begin{equation*}
\dist(f,\widetilde{f})=2^{n}+2^k-2-2\displaystyle\sum_{i=1}^{2^{k-1}+1}|\{\xx\in E_{i}^{\perp}\setminus\{\00\}: f(\xx)=1\}|.
\end{equation*}
\end{theorem}
\begin{proof}
The support of $f$ can be written as $\supp(f)=(\supp(f)\setminus \{\00\})\cup\{\00\}$, where $\supp(f)\setminus\{\00\}=\cup_{i=1}^{2^{k-1}+1}E_{i}\setminus \{\00\}=\cup_{i=1}^{2^{k-1}+1}(E_{i}\setminus \{\00\})$.\newline It is clear that $W_{f}(\00)=\displaystyle\sum_{\xx\in\F_{2}^{n}} (-1)^{f(\xx)}=\displaystyle\sum_{\xx\in\F_{2}^{n}} (1-2f(\xx))=2^{n}-2((2^{k-1}+1)(2^k-1)+1)=-2^k$. From Proposition \ref{metricTh1}, we have
\allowdisplaybreaks[4]
\begin{align*}
\dist(f,\widetilde{f})&=2^{n-1}-2^{k-1}(-1)^{f(\00)}+\frac{1}{2^k}\displaystyle\sum_{\aa\in\supp(f)}W_{f}(\aa)\\
&=2^{n-1}-2^{k-1}(-1)^{f(\00)}+\frac{1}{2^k}\displaystyle\sum_{\aa\in\supp(f)\setminus\{\00\}}W_{f}(\aa) +\frac{1}{2^k}W_{f}(\00)\\
&=2^{n-1}+2^{k-1}+\frac{1}{2^k}\displaystyle\sum_{i=1}^{2^{k-1}+1}\displaystyle\sum_{\aa\in E_{i}\setminus\{\00\}}W_{f}(\aa)+\frac{1}{2^k} (-2^k)\\
&=2^{n-1}+2^{k-1}+\frac{1}{2^k}\displaystyle\sum_{i=1}^{2^{k-1}+1}\displaystyle\sum_{\aa\in E_{i}\setminus\{\00\}}\displaystyle\sum_{\xx\in\F_{2}^{n}} (-1)^{f(\xx) \+ \aa \cdot \xx}-1\\
&=2^{n-1}+2^{k-1}+\frac{1}{2^k}\displaystyle\sum_{i=1}^{2^{k-1}+1}\displaystyle\sum_{\xx\in\F_{2}^{n}} (-1)^{f(\xx)}\displaystyle\sum_{\aa\in E_{i}\setminus\{\00\}} (-1)^{\aa \cdot \xx}-1\\
&=2^{n-1}+2^{k-1}+\frac{1}{2^k}\displaystyle\sum_{i=1}^{2^{k-1}+1}\displaystyle\sum_{\xx\in\F_{2}^{n}} (-1)^{f(\xx)}\bigg(\displaystyle\sum_{\aa\in E_{i}} (-1)^{\aa \cdot \xx}-1\bigg)-1\\
&=2^{n-1}+2^{k-1}+\displaystyle\sum_{i=1}^{2^{k-1}+1}\displaystyle\sum_{\xx\in E_{i}^{\perp}} (-1)^{f(\xx)} +2^{k-1}+1-1\\
&=2^{n-1}+2^k+\displaystyle\sum_{i=1}^{2^{k-1}+1}\displaystyle\sum_{\xx\in E_{i}^{\perp}}(1-2f(\xx))\\
&=2^{n-1}+2^k+(2^{n-1}+2^k)-2\displaystyle\sum_{i=1}^{2^{k-1}+1}|\{\xx\in E_{i}^{\perp}: f(\xx)=1\}|\\
&=2^{n}+2^{k+1}-2\displaystyle\sum_{i=1}^{2^{k-1}+1}|\{\xx\in E_{i}^{\perp}: f(\xx)=1\}|\\
&=2^{n}+2^{k+1}-2\displaystyle\sum_{i=1}^{2^{k-1}+1}(1+|\{\xx\in E_{i}^{\perp}\setminus\{\00\}: f(\xx)=1\}|)\\
&=2^{n}+2^k-2-2\displaystyle\sum_{i=1}^{2^{k-1}+1}|\{\xx\in E_{i}^{\perp}\setminus\{\00\}: f(\xx)=1\}|.
\end{align*}
\end{proof}
\begin{corollary}
\label{thm-dist-ps2}
Let $f\in\mathfrak{B}_n$, $n=2k\geq 4$, be a $\mathcal{PS}^{+}$ bent function and $\widetilde{f}$ be dual of $f$. Then $\dist(f,\widetilde{f})\leq 2^{n}-2^k$.
\end{corollary}
\begin{proof}
Let $f\in\mathfrak{B}_n$, $n=2k\geq 4$, be a $\mathcal{PS}^{+}$ bent function such that $\supp(f)=\cup_{i=1}^{2^{k-1}+1}E_{i}$, where $E_{i}$'s are ``disjoint'' $k$-dimensional subspaces of $\F_{2}^{n}$.  If $\widetilde{f}$ is the dual of $f$, then $|\supp(\widetilde{f})|=2^{n-1}+2^{k-1}=(2^{n-1}-2^{k-1})+2^k$. So, there exist at least $2^k-1$ nonzero elements of $\xx\in \supp(\widetilde{f})$, such that $f(\xx)=1$ as $|\supp(f)|=2^{n-1}+2^{k-1}$. Thus, $\displaystyle\sum_{i=1}^{2^{k-1}+1}|\{\xx\in E_{i}^{\perp}\setminus\{\00\}: f(\xx)=1\}|\geq 2^k-1$, and from Corollary \ref{cor-ps+}, we 
have $\dist(f,\widetilde{f})\leq 2^{n}+2^k-2-2(2^k-1)=2^{n}-2^k$.
\end{proof}
For example let us consider $n=4$ and the $2$-dimensional ``distinct'' subspaces of $\F_{2}^{4}$
\allowdisplaybreaks[4]
\begin{align*}
E_{1}&=\{0000,0001,0100,0101\},~ E_{2}=\{0000,0010,1000,1010\}\\
E_{3}&=\{0000,0011,1101,1110\},~E_{4}=\{0000,0110,1001,1111\}\\
E_{5}&=\{0000,0111,1011,1100\}.
\end{align*}

It is clear that $E_{1}^{\perp}=E_{2}$, $E_{3}^{\perp}=E_{5}$ and $E_{4}^{\perp}=E_{4}$. Suppose $g_1,g_2,g_3\in\mathfrak{B}_4$ are $\mathcal{PS}^{+}$ bent functions such that $\supp(g_1)=E_{1}\cup E_{2}\cup E_{4}$, $\supp(g_2)=E_{1}\cup E_{2}\cup E_{3}$, and $\supp(g_3)=E_{1}\cup E_{3}\cup E_{4}$.  Then $g_1$ is self-dual bent, i.e., $\dist(g_1,\widetilde{g_1})=0$, and
\begin{equation*}
\dist(g_2,\widetilde{g_2})=6, ~\dist(g_3,\widetilde{g_3})=12.
\end{equation*}
\subsection{Distances between $\cP\cS^{-}$ bents and their duals}
Let $f \in \mathfrak{B}_n$ be a bent  function in  $\cP\cS^{-}$ and $E_{i}$'s, $i = 1, \ldots, 2^{k-1}$  are mutually ``disjoint'' $k$-dimensional subspaces of $\F_{2}^{n}$ such that $\supp(f) = \bigcup_{i=1}^{2^{k-1}} E_{i} \setminus \{\00\}$. 
It is straightforward that $f$ is self-dual if and only if 
\begin{equation}
\label{selfDualobvious}
\bigcup_{i=1}^{2^{k-1}} E_{i} = \bigcup_{i=1}^{2^{k-1}} E_{i}^{\perp}. 
\end{equation}

But it is interesting to identify all possible distances between a bent function in $\cP\cS^{-}$ and its dual. We provide a technical result below.
\begin{theorem}
\label{distps}
Let $n=2k\geq 4$ and $f\in\mathfrak{B}_n$ be a $\mathcal{PS}^{-}$ bent function such that $\supp(f)=\cup_{i=1}^{2^{k-1}} E_{i}\setminus \{\00\}$, where $E_{i}$'s, $1\leq i\leq 2^{k-1}$,  are mutually ``disjoint'' $k$-dimensional subspaces of $\F_{2}^{n}$. Then 
\begin{equation*}
\dist(f,\widetilde{f})=2^{n}-2^k-2\displaystyle\sum_{i=1}^{2^{k-1}}|\{\xx\in E_{i}^{\perp}: f(\xx)=1\}|.
\end{equation*}
\end{theorem}
\begin{proof}
For any $i \in \{1, 2, \ldots, 2^{k-1}\}$ we have $$\displaystyle\sum_{\xx\in E_{i}^{\perp}} (-1)^{f(\xx)}=2^k-2|\{\xx\in E_{i}^{\perp}: f(\xx)=1\}|,$$ and $\supp(f)=\cup_{i=1}^{2^{k-1}} E_{i}\setminus \{\00\}=\cup_{i=1}^{2^{k-1}} E_{i}^*$, where $E_{i}^* = E_{i} \setminus \{\00\}$. From Proposition \ref{metricTh1}, we have
\allowdisplaybreaks[4]
\begin{align*}
\dist(f,\widetilde{f})&=2^{n-1}-2^{k-1}(-1)^{f(\00)}+\frac{1}{2^k}\displaystyle\sum_{\uu\in\supp(f)}W_{f}(\uu)\\
&=2^{n-1}-2^{k-1}+\frac{1}{2^k}\displaystyle\sum_{i=1}^{2^{k-1}}\displaystyle\sum_{\uu\in E_{i}^*}\displaystyle\sum_{\xx\in\F_{2}^{n}} (-1)^{f(\xx) \+ \uu\cdot \xx}\\
&=2^{n-1}-2^{k-1}+\frac{1}{2^k}\displaystyle\sum_{i=1}^{2^{k-1}}\displaystyle\sum_{\xx\in\F_{2}^{n}} (-1)^{f(\xx)}\displaystyle\sum_{\uu\in E_{i}^*} (-1)^{\uu\cdot \xx}\\
&=2^{n-1}-2^{k-1}+\frac{1}{2^k}\displaystyle\sum_{i=1}^{2^{k-1}}\displaystyle\sum_{\xx\in\F_{2}^{n}} (-1)^{f(\xx)}\left(\displaystyle\sum_{\uu\in E_{i}} (-1)^{\uu\cdot \xx}-1\right)\\
&=2^{n-1}-2^{k-1}+\displaystyle\sum_{i=1}^{2^{k-1}}\displaystyle\sum_{\xx\in E_{i}^{\perp}} (-1)^{f(\xx)}-\frac{1}{2^k}\displaystyle\sum_{i=1}^{2^{k-1}}\displaystyle\sum_{\xx\in\F_{2}^{n}}(-1)^{f(\xx)}\\
&=2^{n-1}-2^{k-1}+2^{n-1}-2\displaystyle\sum_{i=1}^{2^{k-1}}|\{\xx\in E_{i}^{\perp}: f(\xx)=1\}|-2^{k-1}\\
&=2^{n}-2^k-2\displaystyle\sum_{i=1}^{2^{k-1}}|\{\xx\in E_{i}^{\perp}: f(\xx)=1\}|.
\end{align*}
\end{proof}
\begin{corollary}
\label{thm-dist-ps1}
Let $f\in\mathfrak{B}_n$, $n=2k\geq 4$, be a $\mathcal{PS}^{-}$ bent function with its dual $\widetilde{f}$. Then $\dist(f,\widetilde{f})\leq 2^{n}-2^k$.
\end{corollary}
\begin{proof}
We know that $|\{\xx\in E_{i}^{\perp}: f(\xx)=1\}|\geq 0$ for all $i=1,2,\ldots,2^{k-1}$. Then from Theorem~\ref{distps}, we get the result.
\end{proof}
For example let $E_{i}$, $1\leq i\leq 5$, are subspaces of $\F_{2}^{4}$, defined as above. Suppose $f_1,f_2,f_3$ are three $\mathcal{PS}^{-}$ bent functions in $4$ variables such that $\supp(f_1)=E_{1}\cup E_{2}\setminus\{\00\}$, $\supp(f_2)=E_{3}\cup E_{4}\setminus\{\00\}$ and $\supp(f_3)=E_{1}\cup E_{3}\setminus\{\00\}$. Then $f_1$ is self-dual bent, i.e., $\dist(f_1,\widetilde{f_1})=0$, and
\begin{equation*}
\dist(f_2,\widetilde{f_2})=6, ~\dist(f_3,\widetilde{f_3})=12.
\end{equation*}

\subsubsection{Distance between $\mathcal{PS}_{ap}$ bents and their duals} From corollaries \ref{thm-dist-ps2} and \ref{thm-dist-ps1}, we have the distance between an $n$--variable, $n\geq 4$, bent function in the $\mathcal{PS}$ class and it's dual, is even and bounded above by $2^{n}-2^{\frac{n}{2}}$. Dillon~\cite{Dillon72,Dillon74} constructed a subclass of $\mathcal{PS}$ that is known as $\mathcal{PS}_{ap}$ class. The subscript ``$ap$'' stands for ``affine plane.''

In this section, we obtained an explicit distribution of Rayleigh quotients for the functions in this class. We first go through the construction of Dillon's function by using Desarguessian projective planes.

Let us identify $\F_{2}^{n}$ with $\F_{2^k} \times \F_{2^k}$ where
$n = 2k$. For each $a \in \F_{2^k}$, define
$E_a = \{ (x, xa) : x \in \F_{2^k}\}$.
It is evident that, for each $a \in \F_{2^k}$,  $E_a$ is a $k$-dimensional subspace of 
$\F_{2^k} \times \F_{2^k}$, and $E_a \cap E_b = \{\00\}$
for all $a, b \in \F_{2^k}$ such that $a \neq b$. Let 
$E_{\infty} = \{(0, y) : y \in \F_{2^k}\}$.
Note that $E_{\infty} \cap E_a = \{\00\}$, for all $a \in \F_{2^k}$. 

The Desarguesian spread in $\F_{2^k}\times \F_{2^k}$ is the 
collection
\begin{equation}
\label{dspread}
\mathfrak{D} = \{E_a : a \in \F_{2^k}\} \cup \{E_{\infty}\}.
\end{equation}

For any $\mathfrak{C} \subseteq \mathfrak{D}$ and $\abs{\mathfrak{C}} = 2^{k-1}$, a function of the form
\begin{equation}
\label{psap}
f(\xx) = \displaystyle\sum_{E \in \mathfrak{C}} \11_{E} (\xx) - 2^{k-1} \11_{\{\00\}}(\xx),
\end{equation}
is a bent function in $\mathcal{PS}^{-}$. All such bent functions 
form a subclass denoted by $\mathcal{PS}_{ap}$. The finite field $\F_{2^k}$, considered as a vector space over $\F_{2}$, is endowed with the inner product $(x, y) \mapsto \T_1^k(xy)$. This inner product can be extended to an inner product $\braket{\cdot, \cdot} : (\F_{2^k} \times \F_{2^k} ) \times ( \F_{2^k} \times \F_{2^k}) \rightarrow \F_{2}$ defined by 
\begin{align*}
\braket{(x,y),  (x', y')} = \T_{1}^{k}(x x') + \T_{1}^{k}(yy'). 
\end{align*}
We observe the following result of Desarguesian spread in $\F_{2^{k}}\times \F_{2^{k}}$. We can partition the $\mathcal{PS}_{ap}$ bent functions which are constructed by using Desarguesian spread. Let $\mathcal{PS}_{ap}^{ds}(2k)$ be the set of all bent function in $2k$ variables in $\mathcal{PS}_{ap}$ constructed by using Desarguesian spread $\mathfrak{D}$ defined in \eqref{dspread}. 
Let us define  
\begin{equation*}
cl(i)=\{f\in \mathcal{PS}_{ap}^{ds}(2k): \dist(f,\widetilde{f})=(2^{k+1}-2)(2^{k-1}-i)\},
\end{equation*}
for all $i=0,1,\ldots, 2^{k-1}$. Thus, all self-dual bent functions in $\mathcal{PS}_{ap}^{ds}(2k)$ is $cl(2^{k-1})$. It is clear that $cl(i)\cap cl(j)=\emptyset$ for all $0\leq i\neq j\leq 2^{k-1}$ and $\mathcal{PS}_{ap}^{ds}(2k)=\cup_{i=0}^{2^{k-1}} cl(i)$.

\begin{lemma}
For any $E \in \mathfrak{D}$, its dual $E^{\perp} \in \mathfrak{D}$. 
\end{lemma}
\begin{proof}
First, consider the case when $E = E_a$ for some $a \in \F_{2^k}^*$. 
Let $(y, z) \in E_{a}^{\perp}$, so that 
$$\braket{(y, z), (x, xa)} = 0,$$
for all $x \in \F_{2^k}$. That is 
$$\T_{1}^{k}(yx) + \T_{1}^{k}(zxa) = \T_{1}^{k}((y+za)x) = 0,$$
for all $x \in \F_{2^k}$.\\This implies $y + za = 0$, that is,  
$z  = y a^{-1}$. So $$E_a^{\perp} = \{(y, ya^{-1}) : y \in \F_{2^k}\} = E_{a^{-1}} \in \mathfrak{D}.$$We complete the proof by observing that $E_{0}, E_{\infty} \in \mathfrak{D}$ are mutually orthogonal. 
\end{proof}
The direct count for self-dual $\mathcal{PS}_{ap}$ is given in~\cite[Theorem 4.5]{CarletDPS10}. But our next theorem characterizes self-dual $\mathcal{PS}_{ap}$ bent functions along with their count. Our approach to this problem lets us count the total number of self-dual $\mathcal{PS}_{ap}$ functions as well as obtain the complete distribution of the Rayleigh quotients for functions in $\mathcal{PS}_{ap}$.
\begin{theorem}
\label{selfdualpsap}
Let $\mathfrak{D}$ be the Desarguesian spread defined in 
\eqref{dspread}. A $\mathcal{PS}_{ap}$ class of bent function 
$f: \F_{2^k} \times \F_{2^k} \rightarrow \F_{2}$ defined in \eqref{psap}
is self-dual if and only if
the following conditions are satisfied
\begin{enumerate}
\item $E_{1} \notin \mathfrak{C}$; 
\item $\{E_{0}, E_{\infty}\} \cap\mathfrak{C} 
\in \{\emptyset, \{E_{0}, E_{\infty}\} \}$; 
\item For all $a \in \F_{2^k}^* \setminus\{1\}$, $E_a \in \mathfrak{C}$ if 
and only if  $E_{a^{-1}} \in \mathfrak{C}$.       
\end{enumerate}
The number of self-dual $\mathcal{PS}_{ap}$-bent functions is ${2^{k-1}-1} \choose 2^{k-2}$.  
\end{theorem} 
\begin{proof}
The dual of $f$ defined in \eqref{psap} is 
\begin{equation*}
\widetilde{f}(\xx) = \displaystyle\sum_{E \in \mathfrak{C}} \phi_{E^{\perp}} (\xx) - 2^{k-1} \11_{\{\00\}}(\xx). 
\end{equation*}
Let $\mathfrak{I} = \{ E_a : a \in \F_{2^k} \cup \{\infty\}, a \neq 1\} $. For $f$ to be a self-dual bent function 
$\abs{\mathfrak{C} \cap \mathfrak{I}}$ is even. This is because 
$E_{a} \in \mathfrak{C}$ if and only if $E_{a^{-1}} \in \mathfrak{C}$, 
where $a \in \F_{2^k}^* \setminus \{1\}$, and 
$E_{0} \in \mathfrak{C}$ if and only if $E_{\infty} \in \mathfrak{C}$. 
So the subspaces in $\mathfrak{I}$ appear or do not appear in  pairs. If $E_{1} \in \mathfrak{C}$,  then $\abs{\mathfrak{I}} = 2^{k-1} - 1$.  This is not possible. 

If $f$ is self-dual either both 
$E_{0}, E_{\infty} \in \mathfrak{C}$ or 
$E_{0}, E_{\infty} \notin \mathfrak{C}$. Further,  for all $a \in \F_{2^k}^* \setminus\{1\}$, $E_a \in \mathfrak{C}$ if 
and only if  $E_{a^{-1}} \in \mathfrak{C}$.       
Therefore, the three conditions are satisfied. It is not difficult  to check that these conditions are sufficient for self-duality. 

If $\{E_{0}, E_{\infty}\} \subseteq \mathfrak{C}$, then the  remaining subspaces appear in pairs from the set
$\mathfrak{D} \setminus \{E_{0}, E_{1}, E_{\infty}\}$. 
The number of elements $\abs{\mathfrak{D} \setminus \{E_{0}, E_{1}, E_{\infty}\}}  = 2^{k}-2$. Therefore, the number of subspace pairs
that may be selected from 
 $\mathfrak{D} \setminus \{E_{0}, E_{1}, E_{\infty}\}$ is 
 $2^{k-1} - 1$. To construct $\mathfrak{C}$ 
 we need $2^{k-2}-1$ pair of subspaces, since 
 $\{E_{0}, E_{\infty}\} \subseteq \mathfrak{C}$. 
 This can be done in ${2^{k-1}-1}\choose{2^{k-2} - 1}$
ways. 

In case  $\{E_{0}, E_{\infty}\} \cap \mathfrak{C} = \{\emptyset\}$, 
we have to choose all $2^{k-2}$ pairs in $\mathfrak{C}$ from 
$\mathfrak{D} \setminus \{E_{0}, E_{1}, E_{\infty}\}$. This can be 
done in ${2^{k-1}-1}\choose{2^{k-2}}$. Thus, the total number of 
self-dual $\mathcal{PS}_{ap}$-bent functions is 
\begin{equation*}
\label{countselfdualpsap}
{{2^{k-1}-1}\choose{2^{k-2} - 1}}
+ {{2^{k-1}-1}\choose{2^{k-2}}}
= {{2^{k-1}}\choose{2^{k-2}}}
\end{equation*}
\end{proof}
We can calculate the distribution of the distance between a $\mathcal{PS}_{ap}$ bent function and it's dual when we consider the Desarguesian spread in $\F_{2^k}\times \F_{2^k}$ defined in \eqref{dspread}. 
Let $\mathfrak{C}^{\perp}=\{E^{\perp}: E\in \mathfrak{C}\}$ where $\mathfrak{C}\subseteq \mathfrak{D}$.
\begin{theorem}
Let $\mathfrak{D}$ be the Desarguesian spread defined in \eqref{dspread}. 
A $\mathcal{PS}_{ap}$ bent function  $f: \F_{2^k} \times \F_{2^k} \rightarrow \F_{2}$ defined in \eqref{psap} with its dual $\widetilde{f}$ of the form $\widetilde{f}(\xx) = \displaystyle\sum_{E \in \mathfrak{C}^{\perp}} \11_{E} (\xx) - 2^{k-1} \11_{\{\00\}}(\xx)$.  
\begin{enumerate}
\item If $E_{1} \notin \mathfrak{C}$ and $|\mathfrak{C}\cap\mathfrak{C}^{\perp}|=i$, $0\leq i\leq 2^{k-1}$, then $N_{f} = 2^{n} - 4(2^{k-1} - i) ( 2^k -1)$.  
\item If $E_{1} \in \mathfrak{C}$ and $|\mathfrak{C}\cap\mathfrak{C}^{\perp}|=1+j$, $0\leq j\leq 2^{k-1}-1$, then $N_{f} = 2^{n} - 4(2^{k-1}-j - 1)(2^k-1)$. 
\end{enumerate}
\end{theorem} 
\begin{proof}
Let $E_{1} \notin \mathfrak{C}$ and $\abs{\mathfrak{C} \cap \mathfrak{C}^{\perp}} = i$ where 
$i \in \{0, 1, \cdots, 2^{k-1}\}$. The Hamming distance between $f$ and $\widetilde{f}$ is 
$\dist(f, \widetilde{f}) = 2(2^{k-1}-i)(2^k-1)$. 
Therefore
\begin{equation*}
N_{f} = 2^{n} - 4(2^{k-1} - i) ( 2^k -1). 
\end{equation*}
Let $E_{1} \in \mathfrak{C}$ and $\abs{\mathfrak{C} \cap \mathfrak{C}^{\perp}} = 1+j$ where 
$j \in \{0, 1, \cdots, 2^{k-1}-1\}$. The Hamming distance between $f$ and $\widetilde{f}$ is 
$\dist(f, \widetilde{f}) = 2(2^{k-1}-j - 1)(2^k-1)$. 
Therefore
\begin{equation*}
N_{f} = 2^{n} - 4(2^{k-1}-j - 1)(2^k-1). 
\end{equation*}
\end{proof}
From the above theorem, we get the explicit expression for $N_{f}$ for the Dillon function of $\mathcal{PS}_{ap}$ class. The expression of Rayleigh quotients on $\mathcal{PS}_{ap}$ in Danielsen, Parker, and Sol\'e~\cite[Theorem7]{DanielsenPS09}, also mentioned above is of the form of character sum which is not in explicit form.
From above we have the following observations:
\begin{itemize}
\item A function $f$ in $\mathcal{PS}_{ap}$ in $2k$--varibles, the $n_{f}=2^{n}$, if $i=2^{k-1}$ or $j = 2^{k-1}-1$ i.e. for $i=2^{k-1}$ or $j = 2^{k-1}-1$, the function $f$ is self-dual function.
\item The distance between a $\mathcal{PS}_{ap}$ bent function and it's dual in $2k$ variables defined by Desarguesian spread is multiple of $2^{k+1}-2$. For example, if $k=2$, the only possible such distances are $0,6$ and $12$. For $k=3$, the only possible such distances are $0,14, 28, 42$ and $56$.
\end{itemize}

\begin{proposition}
    There is no anti-self-dual bent function in $\mathcal{PS}_{ap}.$
\end{proposition}
\begin{proof}
We will prove this by the method of contradiction. Let 
 $n=2k$ and $n\geq4$. Let $f \in \mathcal{PS}_{ap}$ be a $n$-variable anti-self-dual bent function such that $E_{1} \notin \mathfrak{C}$ and $|\mathfrak{C}\cap\mathfrak{C}^{\perp}|=i$, $0\leq i\leq 2^{k-1}$, then 
 \begin{equation*}
 N_{f} = 2^{n} - 4(2^{k-1} - i) ( 2^k -1).
 \end{equation*}
Thus, $f$ is a anti-self-dual function $\Rightarrow N_{f} = -2^{n}$, for some $i,~0\leq i\leq 2^{k-1}$.
\begin{align*}
\Rightarrow &2^{n}- 4(2^{k-1} - i) ( 2^k -1)=-2^{n}\\
\Rightarrow &2^{n+1}=4(2^{k-1} - i) ( 2^k -1)\\
\Rightarrow &2^{n-1}=(2^{k-1} - i) ( 2^k -1)
\end{align*} But this is a contradiction, since $n\geq4 \Rightarrow k\geq 2 \Rightarrow ( 2^k -1) > 1$, so, the right-hand side product has an odd factor other than $1$ but the left-hand side has no such factor.\newline Similarly, we can show that the same when $E_{1} \in \mathfrak{C}$.\\ Hence, for any function $f \in \mathcal{PS}_{ap}, N_{f} \neq -2^{n}$.
\end{proof} 
For small $n \geq 4$, and $f \in \mathcal{PS}_{ap}$, we have computed the all possible values of $N_{f}$ and $\dist(f, \widetilde{f})$ which is summarized in the following table:
\allowdisplaybreaks
\begin{center}
\begin{tabular}{|p{0.4cm}| p{7cm}| p{7cm}|} 
 \hline
 $n$ & $N_{f}$ & $\dist(f, \widetilde{f})$ \\
 \hline
 4 & $-8$, 4, 16 & 0, 6, 12 \\ 
 \hline
 6 & $-48$, $-20$, 8, 36, 64 & 0, 14, 28, 42, 56 \\
 \hline
 8 & $-224$, $-164$, $-104$, $-44$, 16, 76, 136, 196, 256 & 0, 30, 60, 90, 120, 150, 180, 210, 240\\  
 \hline
 10 & $-960$, $-836$, $-712$, $-588$, $-464$, $-340$, $-216$, $-92$, 32, 156, 280, 404, 528, 652, 776, 900, 1024 & 0, 62, 124, 186, 248, 310, 372, 434, 496, 558, 620, 682, 744, 806, 868, 930, 992\\
 \hline
 12 & $-3968$, $-3716$, $-3464$, $-3212$, $-2960$, $-2708$, $-2456$, $-2204$, $-1952$, $-1700$, $-1448$, $-1196$, $-944$, $-692$, $-440$, $-188$, 64, 316, 568, 820, 1072, 1324, 1576, 1828, 2080, 2332, 2584, 2836, 3088, 3340, 3592, 3844, 4096 & 0, 126, 252, 378, 504, 630, 756, 882, 1008, 1134, 1260, 1386, 1512, 1638, 1764, 1890, 2016, 2142, 2268, 2394, 2520, 2646, 2772, 2898, 3024, 3150, 3276, 3402, 3528, 3654, 3780, 3906, 4032\\
 \hline
 14 & $-16128$, $-15620$, $-15112$, $-14604$, $-14096$, $-13588$, $-13080$, $-12572$, $-12064$, $-11556$, $-11048$, $-10540$, $-10032$, $-9524$, $-9016$, $-8508$, $-8000$, $-7492$, $-6984$, $-6476$, $-5968$, $-5460$, $-4952$, $-4444$, $-3936$, $-3428$, $-2920$, $-2412$, $-1904$, $-1396$, $-888$, $-380$, 128, 636, 1144, 1652, 2160, 2668, 3176, 3684, 4192, 4700, 5208, 5716, 6224, 6732, 7240, 7748, 8256, 8764, 9272, 9780, 10288, 10796, 11304, 11812, 12320, 12828, 13336, 13844, 14352, 14860, 15368, 15876, 16384 & 0, 254, 508, 762, 1016, 1270, 1524, 1778, 2032, 2286, 2540, 2794, 3048, 3302, 3556, 3810, 4064, 4318, 4572, 4826, 5080, 5334, 5588, 5842, 6096, 6350, 6604, 6858, 7112, 7366, 7620, 7874, 8128, 8382, 8636, 8890, 9144, 9398, 9652, 9906, 10160, 10414, 10668, 10922, 11176, 11430, 11684, 11938, 12192, 12446, 12700, 12954, 13208, 13462, 13716, 13970, 14224, 14478, 14732, 14986, 15240, 15494, 15748, 16002, 16256\\ [1ex] 
 \hline
\end{tabular}
\end{center}

\section{Conclusion}
To the best of our knowledge, it is an open problem to check the distribution of the Rayleigh quotient of bent functions in $\mathcal{PS}$. Our results above demonstrate that we could approach this problem by checking different spreads and partial spreads over finite fields.

\vspace{1em}
\textbf{Acknowledgments.}

The work of Aleksandr Kutsenko is supported by the Mathematical Center in Akademgorodok under agreement No.~075-15-2022-282 with the Ministry of Science and Higher Education of the Russian Federation. The work of Mansi is supported by a fellowship from CSIR, Govt. of India (09/143(1028)/2020-EMR-I).

\end{document}